\documentclass[final,3p,11pt]{elsarticle}

\usepackage{epsfig,amsmath}
\usepackage{bm,graphicx,amssymb,psfrag}
\usepackage{amsfonts}
\usepackage[latin1]{inputenc}

\usepackage{natbib,booktabs}

 \usepackage{amsthm}
 
\usepackage{algorithmic}
\usepackage{algorithm}

\usepackage{hyperref}

\usepackage[usenames]{color}

\usepackage[nolist]{acronym}

\newcommand{\EX}[1] {{\mathbb{E}}\left\{{#1}\right\}}

\newcommand{\W}{\mathcal{TW}}
\newcommand{\X}{\mathcal{G}}
\def\nmin{n_{\text{min}}}
\def\nmax{n_{\text{max}}}

\def\naij{n_{\text{mat}}}

\def\reggamma{P}

\newcommand{\td}[1] {\tilde{#1}}

\newcommand{\bl}[1] {\text{\boldmath ${\lambda}$}_{#1}}

\def\Nmin{n_{\text{min}}}
\def\Nmax{n_{\text{max}}}

\def\Hypergeometric1F1{${}_{1}F_{1}$}

\def\sgn{\text{sgn}}
\def\erf{\text{erf}}
\def\aappx{\alpha}
\def\kappx{{k}}
\def\tappx{\theta}

\newtheorem{corollary}{Corollary}

\newtheorem{theorem}{Theorem}

\newdefinition{remark}{Remark}
\newproof{pf}{Proof}

\journal{Journal of Multivariate Analysis}

\begin{document}

\begin{acronym}
\scriptsize
\acro{AcR}{autocorrelation receiver}
\acro{ACF}{autocorrelation function}
\acro{ADC}{analog-to-digital converter}
\acro{AWGN}{additive white Gaussian noise}
\acro{BCH}{Bose Chaudhuri Hocquenghem}
\acro{BEP}{bit error probability}
\acro{BFC}{block fading channel}
\acro{BPAM}{binary pulse amplitude modulation}
\acro{BPPM}{binary pulse position modulation}
\acro{BPSK}{binary phase shift keying}
\acro{BPZF}{bandpass zonal filter}
\acro{CD}{cooperative diversity}
\acro{CDF}{cumulative distribution function}
\acro{CCDF}{complementary cumulative distribution function}
\acro{CDMA}{code division multiple access}
\acro{c.d.f.}{cumulative distribution function}
\acro{ch.f.}{characteristic function}
\acro{CIR}{channel impulse response}
\acro{CR}{cognitive radio}
\acro{CSI}{channel state information}
\acro{DAA}{detect and avoid}
\acro{DAB}{digital audio broadcasting}
\acro{DS}{direct sequence}
\acro{DS-SS}{direct-sequence spread-spectrum}
\acro{DTR}{differential transmitted-reference}
\acro{DVB-T}{digital video broadcasting\,--\,terrestrial}
\acro{DVB-H}{digital video broadcasting\,--\,handheld}
\acro{ECC}{European Community Commission}
\acro{ELP}{equivalent low-pass}
\acro{FCC}{Federal Communications Commission}
\acro{FEC}{forward error correction}
\acro{FFT}{fast Fourier transform}
\acro{FH}{frequency-hopping}
\acro{FH-SS}{frequency-hopping spread-spectrum}
\acro{GA}{Gaussian approximation}
\acro{GPS}{Global Positioning System}
\acro{HAP}{high altitude platform}
\acro{i.i.d.}{independent, identically distributed}
\acro{IFFT}{inverse fast Fourier transform}
\acro{IR}{impulse radio}
\acro{ISI}{intersymbol interference}
\acro{LEO}{low earth orbit}
\acro{LOS}{line-of-sight}
\acro{BSC}{binary symmetric channel}
\acro{MB}{multiband}
\acro{MC}{multicarrier}
\acro{MF}{matched filter}
\acro{m.g.f.}{moment generating function}
\acro{MI}{mutual information}
\acro{MIMO}{multiple-input multiple-output}
\acro{MISO}{multiple-input single-output}
\acro{MRC}{maximal ratio combiner}
\acro{MMSE}{minimum mean-square error}
\acro{M-QAM}{$M$-ary quadrature amplitude modulation}
\acro{M-PSK}{$M$-ary phase shift keying}
\acro{MUI}{multi-user interference}
\acro{NB}{narrowband}
\acro{NBI}{narrowband interference}
\acro{NLOS}{non-line-of-sight}
\acro{NTIA}{National Telecommunications and Information Administration}
\acro{OC}{optimum combining}
\acro{OFDM}{orthogonal frequency-division multiplexing}
\acro{p.d.f.}{probability distribution function}
\acro{PAM}{pulse amplitude modulation}
\acro{PAR}{peak-to-average ratio}
\acro{PDP}{power dispersion profile}
\acro{p.m.f.}{probability mass function}
\acro{PN}{pseudo-noise}
\acro{PPM}{pulse position modulation}
\acro{PRake}{Partial Rake}
\acro{PSD}{power spectral density}
\acro{PSK}{phase shift keying}
\acro{QAM}{quadrature amplitude modulation}
\acro{QPSK}{quadrature phase shift keying}
\acro{r.v.}{random variable}
\acro{R.V.}{random vector}
\acro{SEP}{symbol error probability}
\acro{SIMO}{single-input multiple-output}
\acro{SIR}{signal-to-interference ratio}
\acro{SISO}{single-input single-output}
\acro{SINR}{signal-to-interference plus noise ratio}
\acro{SNR}{signal-to-noise ratio}
\acro{SS}{spread spectrum}
\acro{TH}{time-hopping}
\acro{ToA}{time-of-arrival}
\acro{TR}{transmitted-reference}
\acro{UAV}{unmanned aerial vehicle}
\acro{UWB}{ultrawide band}
\acro{UWBcap}[UWB]{Ultrawide band}
\acro{WLAN}{wireless local area network}
\acro{WMAN}{wireless metropolitan area network}
\acro{WPAN}{wireless personal area network}
\acro{WSN}{wireless sensor network}
\acro{WSS}{wide-sense stationary}

\acro{SW}{sync word}
\acro{FS}{frame synchronization}
\acro{BSC}{binary symmetric channels}
\acro{LRT}{likelihood ratio test}
\acro{GLRT}{generalized likelihood ratio test}
\acro{LLRT}{log-likelihood ratio test}
\acro{$P_{EM}$}{probability of emulation, or false alarm}
\acro{$P_{MD}$}{probability of missed detection}
\acro{ROC}{receiver operating characteristic}
\acro{AUB}{asymptotic union bound}
\acro{RDL}{"random data limit"}
\acro{PSEP}{pairwise synchronization error probability}

\acro{SCM}{sample covariance matrix}

\acro{PCA}{principal component analysis}

\end{acronym}

\begin{frontmatter}
\title{Distribution of the largest eigenvalue for real Wishart and Gaussian random matrices and a simple approximation for the Tracy-Widom distribution
}
\author{Marco~Chiani}
\address{ 
    DEI, University of Bologna\\
    V.le Risorgimento 2, 40136 Bologna, ITALY\\
}

\begin{abstract}
We  derive efficient recursive formulas giving the exact distribution of the largest eigenvalue for finite dimensional real Wishart matrices and for the Gaussian Orthogonal Ensemble (GOE). In comparing the exact distribution with the limiting distribution  of large random matrices,  we also found that the Tracy-Widom law can be approximated by a properly scaled and shifted gamma distribution, with great accuracy for the values  of common interest in statistical applications.  
\end{abstract}

\end{frontmatter}

\section{Introduction}
\label{sec:intro}

The distribution of the largest eigenvalue of Wishart and Gaussian random matrices plays an important role in many fields of multivariate analysis, including principal component analysis, analysis of large data sets, communication theory and mathematical physics \cite{And:03,Mui:82}. 

The exact \ac{CDF} of the largest eigenvalue for complex finite dimensional central Wishart matrices is given in \cite{Kha:64} in the uncorrelated case  (i.e., with identity covariance), and in \cite{Kha:69} for the correlated case. The extension to non-central uncorrelated complex Wishart is derived in \cite{KanAlo:03}, while the case of double-correlation has been studied in  \cite{MckGraCol:07}. These results can be extended to the case of  covariance matrix having eigenvalues of arbitrary multiplicities by following the approach in  \cite{ChiWinShi:J10}. 

Little is known about the case of real matrices (real Wishart and real Gaussian Orthogonal Ensemble (GOE)), for which hypergeometric functions of a matrix argument {should} be computed. Numerical methods and approximations are provided for instance in \cite{But:02,But:11,Has:12}. These methods are suitable for the analysis of small dimension matrices, allowing the numerical computation of the largest eigenvalue distribution for uncorrelated (see e.g. \cite{But:11}) and correlated (see e.g. \cite{Has:12}) real Wishart matrices. However, formulas for large matrices appear to be unavailable, even for the uncorrelated case, and asymptotic distributions have been used recently as the only alternative to simulation \cite{Joh:01,JohMa:12, Ma:12}. 

In this paper we derive simple expressions for the exact distribution  of the largest eigenvalue for real Wishart matrices with identity covariance and for the GOE, as well as  efficient recursive methods for their numerical computation. For instance, the exact \ac{CDF} of the largest eigenvalue for a  $50$-variate, $50$-degrees of freedom real Wishart distribution  is computed in less than $0.1$ seconds, while about $70$ seconds are required for a $500$-variate,  $500$-degrees of freedom matrix.  So, for example, all results obtained by simulation in \cite[Table I ]{Joh:01} can be easily computed by the exact distribution given here. As a consequence, for problems with quite large matrices (e.g., of dimension up to $500$) the exact distribution of the largest eigenvalue is easily computable, and there is no need for asymptotic approximations.

For larger matrices, instead of computing the exact distributions we can resort to approximations based on the Tracy-Widom distribution. This distribution arises in many fields as the limiting distribution of the largest eigenvalue of large random matrices, with applications including principal component analysis, analysis of large data sets, communication theory and mathematical physics \cite{TraWid:94,TraWid:96,Joha:00,Joh:01,Elk:07,TraWid:09}.  
While its computation requires the numerical solution of the Painlev\'{e} II differential equation \cite{TraWid:94} or the numerical approximation of a Fredholm determinant \cite{Bor:10}, we here show that the Tracy-Widom law can be simply approximated by a properly scaled and shifted gamma distribution, with great accuracy for the values  of common interest in statistical applications.  
This allows to find in closed form the parameters of a shifted gamma distribution approximating the largest eigenvalue distribution for the Wishart and Gaussian matrices, both in the real and in the complex case.
\medskip

\noindent
The novel contributions of this paper are the following:
\begin{itemize}
\item Theorem 1: exact expressions for the distribution of the largest eigenvalue for real Wishart matrices with identity covariance, with an efficient computation method based on recursive formulas; 
\item Theorem 2: exact expressions for the distribution of the largest eigenvalue for GOE matrices, with an efficient computation method based on recursive formulas;
\item A simple and accurate approximation of the Tracy-Widom law based on a scaled and shifted gamma distribution (equation \eqref{eq:appxmc}).
\end{itemize}
Although the main focus is on real matrices, for completeness, besides recalling the known distribution for the complex Wishart case (Theorem 3), we also give a new result for the Gaussian Unitary Ensemble (GUE) in Theorem 4.

\medskip
Throughout the paper we indicate with $\Gamma(.)$ the gamma function, with $\gamma\left(a,x\right)=\int_{0}^{x} t^{a-1} e^{-t} dt$ the lower incomplete gamma function, with  $\reggamma(a,x)=\frac{1}{\Gamma(a)}\gamma(a,x)$ the regularized lower incomplete gamma function, {and with $|\cdot|$ the determinant}.

\section{Exact distribution of the eigenvalues for finite dimensional Wishart and Gaussian symmetric matrices}
\label{sec:exact}

In this section we derive new, simple expressions for the exact distribution of the largest eigenvalue for finite dimensional real Wishart matrices and real symmetric Gaussian matrices. 
We show that the new expressions can be efficiently evaluated even for quite large matrices.
We then analyze the case of complex matrices.

\subsection{Real random matrices: uncorrelated Wishart and the Gaussian Orthogonal Ensemble (GOE)}

Assume a Gaussian real $p \times m$ matrix ${\bf X}$ with \ac{i.i.d.} columns, each with zero mean and identity covariance ${\bf \Sigma=I}$. 
Denoting $\Nmin=\min\{m,p\}$, $\Nmax=\max\{m,p\}$, ${\Gamma}_{m}(a)=\pi^{m (m-1)/4}
        \prod_{i=1}^{m}\Gamma(a-(i-1)/2)$, the joint \ac{p.d.f.} of the (real) ordered eigenvalues $\lambda_1 \geq \lambda_2 \ldots \geq
\lambda_{\Nmin} \geq 0$ of the real Wishart matrix ${\bf W =X X}^T$  is given by \cite{Jam:64, And:03} 
\begin{equation}\label{eq:jpdfuncorrnodet}
f_{\bl{}} (x_{1}, \ldots, x_{\Nmin}) = K \,
     \prod_{i=1}^{\Nmin}e^{-x_{i}/2}x_{i}^{\alpha}
       \prod_{i<j}^{\Nmin}
    \left(x_{i}-x_{j}\right) 
\end{equation}
where {$x_1 \geq x_2 \geq \cdots \geq x_{\Nmin} \geq 0$,  $\alpha \triangleq (\Nmax-\Nmin-1)/2$}, and $K$ is a normalizing constant given by
\begin{equation}
K = \frac{\pi^{\Nmin^2 /2}}
        {2^{\Nmin \Nmax /2} \Gamma_{\Nmin}(\Nmax/2)  \Gamma_{\Nmin}(\Nmin/2)} \, .
\label{eq:K}
\end{equation}

Similarly, for the Gaussian Orthogonal Ensemble the interest is in the distribution of the (real) eigenvalues for real $n \times n$ symmetric matrices whose entries are \ac{i.i.d.} Gaussian  ${\mathcal{N}}(0, 1/2)$ on the upper-triangle, and  \ac{i.i.d.} ${\mathcal{N}}(0, 1)$ on the diagonal \cite{TraWid:09}. 
Their joint \ac{p.d.f.} is \cite{Meh:91,TraWid:09}
\begin{equation} \label{eq:jpdfuncorrGOE}
f_{\bl{}} (x_{1}, \ldots, x_{n}) = K_{GOE} 
     \prod_{i=1}^{n}e^{-x_{i}^2/2} \prod_{i<j}^{n}
    \left(x_{i}-x_{j}\right) 
\end{equation}
where {$x_1 \geq x_2 \geq \cdots \geq x_n$} and $K_{GOE}=[2^{n/2} \prod_{i=1}^{n} \Gamma(i/2)]^{-1}$ is a normalizing constant. Note that the eigenvalues here are distributed over all the reals.

\medskip

\noindent
The following is a new Theorem for real Wishart matrices with identity covariance.

\begin{theorem}
\label{th:cdfwishartreal}
The \ac{CDF} of the largest eigenvalue of the real Wishart matrix ${\bf W}$ is
\begin{equation}
\label{eq:cdfwishartreal}
F_{\lambda_1}(x_1)=\Pr\left\{\lambda_1 \leq x_1\right\}=K' \,   \sqrt{\left|{\bf A}(x_1)\right|}
\end{equation}
with the constant
$$
K' = K \, 2^{\alpha \naij+\naij (\naij+1)/2} \prod_{k=1}^{\naij} \Gamma\left(\alpha+k\right)
$$
{where $\naij=\nmin$ when $\nmin$ is even, and $\naij=\nmin+1$ when $\nmin$ is odd.}

{In \eqref{eq:cdfwishartreal},} when $\nmin$ is even the elements of the $\nmin \times \nmin$ skew-symmetric matrix ${\bf A}(x_1)$  are
\begin{equation}
\label{eq:aij}
a_{i,j}(x_1)= 
I\left(\alpha_j,\alpha_i; \frac{x_1}{2}\right)-I\left(\alpha_i,\alpha_j; \frac{x_1}{2}\right) \qquad i,j=1,\ldots,\nmin
\end{equation}
with $\alpha_i \triangleq \alpha+i=(\Nmax-\Nmin-1)/2+i$, and
\begin{equation}
\label{eq:I}
I(a,b; x)\triangleq \frac{1}{\Gamma(a)} \int_{0}^{x} t^{a-1} e^{-t} \, \reggamma \left(b,t \right)dt .
\end{equation}
When $\nmin$ is odd, the elements of the $(\nmin+1) \times (\nmin+1)$ skew-symmetric matrix ${\bf A}(x_1)$ are as in \eqref{eq:aij}, with the additional elements
\begin{eqnarray}  \label{eq:aijodd}
  a_{i,\nmin+1}(x_1)&=&\frac{2^{-\alpha_{\nmin+1}}}{\Gamma(\alpha_{\nmin+1})} \reggamma \left(\alpha_i,{x_1}/{2} \right)  \qquad i=1, \ldots, \nmin  \\
\label{eq:aijodd1} a_{\nmin+1,j}(x_1)&=&-a_{j,\nmin+1}(x_1)  \qquad\qquad j=1, \ldots, \nmin 
\\   
a_{\nmin+1,\nmin+1}(x_1)&=&0 
\end{eqnarray}

Note that $a_{i,j}(x_1)=-a_{j,i}(x_1)$ and $a_{i,i}(x_1)=0$.

\medskip
{Moreover, the elements $a_{i,j}(x_1)$ can be computed iteratively, without numerical integration or series expansion.}

\end{theorem}

\begin{proof}
{Denoting $\xi(w)=e^{-w/2} w^{\alpha}$, ${\bf w} = \left[w_{1}, w_{2}, \ldots, w_{\Nmin} \right]$, and with ${\bf V}({\bf w})=\left\{w_j^{i-1}\right\}$ the Vandermonde matrix,   
  we have 
\begin{equation}
\label{eq:freversed}
 f_{\bl{}} (w_{\Nmin}, \ldots, w_{1})=K  \prod_{i<j} (w_j-w_i) \prod_{i=1}^{\Nmin} \xi(w_i)= K  \left|{\bf V}({\bf w})\right| \prod_{i=1}^{\Nmin} \xi\left(w_i\right) 
 \end{equation}
 with  $0 \leq w_1 \leq \cdots \leq w_{\Nmin}$ in ascending order.  }
 
The \ac{CDF} of the largest eigenvalue is then 
\begin{eqnarray}
F_{\lambda_1}(x_1)&=&\underset{{0 \leq w_{1} < \ldots < w_{\Nmin} \leq x_1} }{\int\ldots\int} f_{\bl{}} (w_{\Nmin}, \ldots, w_{1}) d{\bf w} \\
&=&K \underset{{0 \leq w_1 < \ldots < w_{\Nmin} \leq x_1}}{\int\ldots\int}   \left|{\bf V}({\bf w})\right|
     \prod_{i=1}^{\Nmin} \xi\left(w_i\right)  d{\bf w} \,. \label{eq:th1proof}
\end{eqnarray}

To evaluate this integral we recall that for a generic $m \times m$ matrix ${\bf \Phi}({\bf w})$ with elements $\left\{\Phi_i(w_j)\right\}$ {where the $\Phi_i(x), \, i=1, \ldots, m$ are generic functions,} the following identity holds \cite{Deb:55,TraWid:98}  
\begin{equation}
\label{eq:debru}
\underset{{a \leq w_1 < \ldots < w_{m} \leq b} }{\int\ldots\int}  \left|{\bf \Phi}({\bf w})\right|  d{\bf w} = \text{Pf}\left({\bf A}\right)
\end{equation}
where $\text{Pf}\left({\bf A}\right)$ is the Pfaffian, and the skew-symmetric matrix $\bf A$ is $m \times m$ for $m$ even, and $(m+1) \times (m+1)$ for $m$ odd, with 
\begin{equation*}
\label{eq:aijdebru}
a_{i,j}=\int_a^b \int_a^b \sgn(y-x) \Phi_i(x) \Phi_j(y) dx dy \qquad i,j = 1, \ldots, m .
\end{equation*}
For $m$ odd the additional elements are $a_{i,m+1}=-a_{m+1,i}=\int_a^b \Phi_i(x) dx$, $i=1, \ldots, m$, and $a_{m+1,m+1}=0$. 

{We recall that, for a skew-symmetric matrix ${\bf A}$, $\left(\text{Pf}\left({\bf A}\right)\right)^2=\left|{\bf A}\right|$ \cite{Deb:55}. 
} 
Then, using \eqref{eq:debru} in \eqref{eq:th1proof} with $a=0, b=x_1$, and $\Phi_i(x)=x^{i-1} \xi(x)=x^{i-1} e^{-x/2}x^{\alpha}$, {after} simple manipulations we get \eqref{eq:cdfwishartreal}. 

Theorem \ref{th:cdfwishartreal} can be used for an efficient computation of the exact \ac{CDF} of the largest eigenvalue for real Wishart matrices, {without numerical integration or infinite series. }
In fact, first we observe that  for {an integer $n$} we have  \cite[Ch. 6]{AbrSte:B70} 
\begin{equation}
\label{eq:reggammarecursive}
\reggamma(a+n,x)=\reggamma(a,x)- e^{-x} \sum_{k=0}^{n-1} \frac{x^{a+k}}{ \Gamma(a+k+1)}.
\end{equation}
Therefore, $\reggamma(a+n,x)$ can be written in closed form when $a=0$ or $a=1/2$, starting from  
$\reggamma(0,x)=1$ and $\reggamma(1/2,x)= \erf\sqrt{x}$. Thus, the elements of the matrix in  \eqref{eq:cdfwishartreal} can be evaluated iteratively without any numerical integration by using the following identities which can be easily verified:
\begin{eqnarray}
I(a,a;x)&=& \frac{1}{2} \reggamma \left(a,x \right)^2  \label{eq:rec1} \\
I(a,b+1; x)&=&I(a,b; x)-\frac{2^{-(a+b)} \Gamma(a+b)}{\Gamma(a)\Gamma(b+1)} \reggamma \left(a+b,2x \right) \label{eq:rec2} \\
I(b,a; x)&=& \reggamma \left(a,x \right) \reggamma \left(b,x \right) -I(a,b; x)  \label{eq:rec3}
\end{eqnarray}
\end{proof}

In summary, the \ac{CDF} in \eqref{eq:cdfwishartreal} is simply obtained, without any numerical integral, by Algorithm 1 reported below, {where $P(a,x)$ can be computed by using \eqref{eq:reggammarecursive}}.
\medskip
\begin{algorithm}[!h]
\renewcommand{\algorithmicrequire}{\textbf{Input:}}
\renewcommand{\algorithmicensure}{\textbf{Output:}}
\caption{\ac{CDF} of the largest eigenvalue of real Wishart matrices}
\begin{algorithmic}[0]
\REQUIRE $\nmin, \nmax, x$
\ENSURE $F_{\lambda_1}(x)=\Pr\left\{\lambda_1 \leq x\right\}$
 \STATE ${\bf A}={\bf 0}$, $\alpha=(\nmax-\nmin-1)/2$
{ \STATE  {Pre-compute \;} $p_\ell=P(\alpha+\ell,x/2)$ and $\gamma_\ell=\Gamma(\alpha+\ell)$ {\; for} $\ell=1 \rightarrow \nmin$
 \STATE  {Pre-compute \;} $q_\ell=2^{-(2\alpha+\ell)} \Gamma(2\alpha+\ell) P(2\alpha+\ell,x)$ {\; for} $\ell=2 \rightarrow 2\nmin-1$
}   \FOR{$i = 1 \to \nmin$}
  	\STATE $b = p_i^2/2$ 
  	\FOR {$j = i \to \nmin-1$}
  		\STATE $b=b-{q_{i+j}}/ ({\gamma_i \gamma_{j+1}}) $ 
		\STATE $a_{i,j+1}=p_i p_{j+1}-2 b$ 
  	\ENDFOR
  \ENDFOR
  \IF{$\nmin$ is odd}
     \STATE Append to ${\bf A}$ one column according to \eqref{eq:aijodd} and a zero row 
  \ENDIF
  \STATE ${\bf A}={\bf A}-{\bf A}^T$
   \RETURN $F_{\lambda_1}(x)=  K' \, \sqrt{|{\bf A}|}$
\end{algorithmic}
\end{algorithm}

\noindent
Implementing directly the algorithm in Mathematica$^{\scriptsize\textregistered}$ on a desktop computer,\footnote{{Mathematica$^{\scriptsize\textregistered}$ version 9, processor clock rate 1.8 GHz.}} for each value $x_1$ we obtain the exact \ac{CDF} in \eqref{eq:cdfwishartreal} for $\nmin=\nmax=50$ in less than $0.1$ second,  for $\nmin=\nmax=200$ in around $5$ seconds, and for $\nmin=\nmax=500$ in around $70$ seconds, {with a computational complexity dominated by the evaluation of the determinant, $\mathcal{O}({\nmin^3})$}. For comparison, the approach based on series expansions of the hypergeometric function of a matrix argument requires hours for matrices with $\nmin=\nmax=50$ \cite{But:11}. 

So, Algorithm 1 {allows}, for example, to exactly evaluate in a few seconds all results investigated by simulation in \cite[Table 1]{Joh:01} and \cite[Tables 1,2]{Elk:03}.

{Finally we mention that, during the revision of this work, it has been observed that by combining \eqref{eq:reggammarecursive}, \eqref{eq:rec2}, \eqref{eq:rec3} and \eqref{eq:aij} the algorithm can be equivalently reformulated starting from $a_{i,i}=0$ with the iteration  
\begin{equation}\label{eq:aij+1}
a_{i,j+1}=a_{i,j}-p_i r_j +  2 {q_{i+j}} / ({\gamma_i \gamma_{j+1}}) \qquad j=i, \ldots, \nmin-1
\end{equation}
with $p_\ell, q_\ell$ as in Algorithm 1 and $r_\ell={e^{-x/2} (x/2)^{\alpha+\ell}}/{\Gamma(\alpha+1+\ell)}$. 
}

\medskip
\noindent
The next are two new corollaries derived from Theorem \ref{th:cdfwishartreal}.
  
\begin{corollary}
\label{cor:realwishart}
The largest eigenvalue of a real Wishart matrix is a mixture of gamma distributions when $\nmax-\nmin$ is odd.
\end{corollary}
\begin{proof}
When $\nmax-\nmin$  is odd, $\alpha=(\nmax-\nmin-1)/2$ is an integer. Then, from \eqref{eq:reggammarecursive} it results that $\reggamma(\alpha+\ell,x)$  is $1$ plus a combination of terms $x^k e^{- x}$. From \eqref{eq:rec1}, \eqref{eq:rec2} and \eqref{eq:rec3} it results therefore that the $I(\alpha_i,\alpha_j;x_1/2)$ in \eqref{eq:aij} is a constant plus a combination (with both positive and negative weights) of terms $x_1^k e^{- \delta x_1}$. Since the Pfaffian can always be written as a polynomial in the matrix entries \cite[eq. (3.1)]{Deb:55}, it results that the \ac{CDF} in \eqref{eq:cdfwishartreal} is of the same type, and its derivative is a mixture of gamma distributions (with both positive and negative weights).
\end{proof}
\begin{corollary}
\label{cor:realwisharteven}
When $\nmax-\nmin$  is even, the distribution 
 in \eqref{eq:cdfwishartreal} is a combination, with both positive and negative weights, of a constant plus terms of the form $x_1^\xi e^{- \delta x_1}$ and terms of the form $x_1^\xi e^{- \delta x_1} \erf\sqrt{x_1/2}$.
\end{corollary}
\begin{proof}
When $\nmax-\nmin$  is even,  $\alpha=(\nmax-\nmin-1)/2$ is of the form $1/2+m$ where $m$ is an integer. Therefore, $\reggamma(\alpha_i,x)$  is $ \erf\sqrt{x}$ plus a combination of a constant and terms $x^\xi e^{- x}$, while $\reggamma(\alpha_i+\alpha_j,x)$  is a combination of terms $x^\xi e^{- x}$. 
Then, the Corollary is proved due to \eqref{eq:rec1}, \eqref{eq:rec2}, \eqref{eq:rec3} and \eqref{eq:aij}. 
\end{proof}

\noindent
The following is a new Theorem for the GOE matrices.

\begin{theorem}
\label{th:cdfGOE}
The \ac{CDF} of the largest eigenvalue for the Gaussian Orthogonal Ensemble (GOE) matrices is
%
%
\begin{equation}
\label{eq:cdfGOE}
F_{\lambda_1}(x_1)=\Pr\left\{\lambda_1 \leq x_1\right\}=K'_{GOE} \sqrt{\left|{\bf A}(x_1)\right|}
\end{equation}
with the constant
$$
K'_{GOE} = K_{GOE} \prod_{k=1}^{\naij} \Gamma\left({k}/{2}\right) .
$$
When $n$ is even, $\naij=n$, the elements of the $n \times n$ skew-symmetric matrix ${\bf A}(x_1)$ are
\begin{equation}
\label{eq:aijGOE}
a_{i,j}(x_1)=  I_{G}(j,i;x_1) - I_{G}(i,j;x_1) 
\end{equation}
with
\begin{equation}
\label{eq:IG}
I_{G}(i,j; x)\triangleq \frac{1}{\Gamma(i/2)} \int_{-\infty}^{x} t^{i-1} e^{-t^2/2} \, \psi\left(j,t \right)dt 
\end{equation}
and 
$$\psi(j,x)\triangleq \frac{1}{\Gamma(j/2)} \int_{-\infty}^{x} t^{j-1} e^{-t^2/2} dt=2^{\frac{j}{2}-1} \left(\text{sgn}(x)^j \reggamma\left(\frac{j}{2},\frac{x^2}{2}\right)-(-1)^j \right).$$
%

When $n$ is odd, $\naij=n+1$ and the elements of the $(n+1) \times (n+1)$ skew-symmetric matrix ${\bf A}(x_1)$ are as in \eqref{eq:aijGOE}, with the additional elements
\begin{eqnarray}
\label{eq:aijoddGOE}
a_{i,n+1}(x_1)&=&\frac{1}{\Gamma((n+1)/2)} \psi(i,x_1)   \qquad i=1, \ldots, n \nonumber \\
a_{n+1,j}(x_1)&=&-a_{j,n+1}(x_1)  \qquad\qquad j=1, \ldots, n  \nonumber \\   
a_{n+1,n+1}(x_1)&=&0. \nonumber
\end{eqnarray}
%

\medskip
Moreover, the elements $a_{i,j}(x_1)$ can be computed iteratively, starting from the gamma function, without numerical integration or series expansion.
\end{theorem}
\begin{proof}
Starting from \eqref{eq:jpdfuncorrGOE} the proof is similar to that for Theorem \ref{th:cdfwishartreal}. 
The elements in \eqref{eq:aijGOE} can be efficiently derived recursively without numerical integration, by using \eqref{eq:rec1}, \eqref{eq:rec2}, and \eqref{eq:rec3},  with the additional relations
\begin{eqnarray}
I_{G}(i,i; x)&=&\psi(i,x)^2/2 \nonumber \\ 
I_{G}(i,j+1; x)&=&2 I_{G}(i,j-1; x) -\frac{\Gamma\left(\frac{i+j-1}{2}\right)}{2 \ \Gamma\left(\frac{i}{2}\right)\Gamma\left(\frac{j+1}{2}\right)} \nonumber \\
&& \times  \left[(-1)^{i+j} + P\left(\frac{i+j-1}{2},x^2\right)\right]  \nonumber \\
I_{G}(i,j; x)&=&\psi(j,0)\psi(i,x)+2^{\frac{i+j}{2}-2} \nonumber \\
&& \times \left[ (-1)^{i+j+1} I\left(\frac{i}{2},\frac{j}{2}; \infty\right) + \sgn(x)^{i+j} I\left(\frac{i}{2},\frac{j}{2}; \frac{x^2}{2}\right)\right]  \nonumber \\
I_{G}(j,i; x)&=&\psi(i,x)\psi(j,x)-I_{G}(i,j; x) \nonumber 
\end{eqnarray}
and $I(a,0,x)=\reggamma(a,x), \,\,\, I(a,0,\infty)=1, \,\,\, I(a,a,\infty)=1/2$.
\end{proof}
\medskip
\noindent
For example, we obtained the exact \ac{CDF} in \eqref{eq:cdfGOE} for $n=100$ in less than $2$ seconds, for $n=200$ in less than $9$ seconds, and for $n=500$ in less than $120$ seconds. 
We are not aware of other efficient methods in literature for computing the exact \ac{CDF} of GOE matrices.

\subsection{Complex random matrices: uncorrelated Wishart and the Gaussian Unitary Ensemble (GUE)}

Assume now a Gaussian complex $p \times m$ matrix ${\bf X}$ with \ac{i.i.d.} columns, each circularly symmetric with zero mean and covariance ${\bf \Sigma}$. The distribution of the (real) ordered eigenvalues 
 of the complex Wishart matrix ${\bf W =X X}^H$ is known since many years from \cite{Jam:64} in terms of hypergeometric functions of matrix arguments. Unfortunately, the expressions given in \cite{Jam:64} are not easy to use, due to the difficulties in evaluating zonal polynomials. 
The first expression of practical usage for the joint distribution of the eigenvalues of a complex Wishart matrix with correlation has been given in \cite{ChiWinZan:J03} by expressing the hypergeometric function of matrix arguments as product of determinants of matrices. More recently, that approach has been expanded to cover the case where ${\bf \Sigma}$ has eigenvalues of arbitrary multiplicity, and to find several statistics regarding the marginal eigenvalues distribution \cite{ChiZan:C08,ChiWinShi:J10,ZanChiWin:J09}. 
By using these approaches, the exact statistics of an arbitrary subset of the ordered eigenvalues can be evaluated easily for finite dimensional complex quadratic forms and Wishart (uncorrelated and correlated) matrices. 

Regarding the largest eigenvalue statistics, below we report a known result for the particular case of uncorrelated complex Wishart matrices (i.e., for ${\bf\Sigma=I}$).
\begin{theorem}
\label{th:cdfwishartcomplex}
The \ac{CDF} of the largest eigenvalue of the uncorrelated complex Wishart matrix ${\bf W}$ is \cite{Kha:64}
\begin{equation}
\label{eq:cdfwishartcomplex}
F_{\lambda_1}(x_1)=\Pr\left\{\lambda_1 \leq x_1\right\}=K_C \left|{\bf A}(x_1)\right|
\end{equation}
where the elements of the $\nmin \times \nmin$ matrix ${\bf A}(x_1)$ are
\begin{equation}
\label{eq:aijcomplex}
a_{i,j}(x_1)=  \int_{0}^{x_1} t^{\nmax-\nmin+i+j-2} e^{-t} dt= \gamma\left(\nmax-\nmin+i+j-1,x_1 \right)  
\end{equation}
and $K_C$ is a normalizing constant given by
\begin{equation}
K_C = \frac{\pi^{\Nmin(\Nmin-1)}}
        {\td{\Gamma}_{\Nmin}(\Nmax)\td{\Gamma}_{\Nmin}(\Nmin)} \,
\label{eq:KC}
\end{equation}
with
$ \td{\Gamma}_{m}(n)=\pi^{m(m-1)/2}
        \prod_{i=1}^{m}(n-i)! \,.
$

\end{theorem}
\begin{corollary}
\label{cor:complexwishart}
The largest eigenvalue of a complex Wishart matrix  is a mixture of gamma distributions.
\end{corollary}
\begin{proof}
From \eqref{eq:reggammarecursive} it results that each element in \eqref{eq:aijcomplex}  is in the form $\Gamma(\ell) (1-e^{-x_1} P_\ell(x_1))$ where $P_\ell(x_1)$ is a degree $\ell-1$ polynomial. 
 Thus, the determinant in  \eqref{eq:cdfwishartcomplex} is a constant plus a combination, with both positive and negative weights,  of terms $x_1^k e^{-\delta x_1}$. 
  Its derivative is therefore a weighted sum of terms  $x_1^k e^{-\delta x_1}$.
\end{proof}

The proof of the above Corollary, simply derived from the (known) Theorem \ref{th:cdfwishartcomplex}, is new. However, the fact that the \ac{p.d.f.} of the largest eigenvalue of complex Wishart matrices can be expressed as a combination of gamma is known (see e.g. \cite{DigMalJam:J03} and the discussion in \cite{ZanChi:J12}). 
Moreover, since the \ac{CDF} for complex Wishart matrices with correlation has a form similar to \eqref{eq:aijcomplex} (see e.g. \cite[eq. (36)]{Kha:69} and  \cite[eq. (1)]{MckGraCol:07}), the Corollary can be easily proved to hold also in the presence of correlation.

Following a more complicated method it is possible to show that each eigenvalue (not just the largest) of a complex Wishart matrix  is a mixture of gamma distributions \cite[Th.1]{ZanChi:J12}. 

\noindent
{We next study} complex Hermitian random matrices with \ac{i.i.d.} ${\mathcal{CN}}(0, 1/2)$  entries on the upper-triangle, and ${\mathcal{N}}(0, 1/2)$ on the diagonal. These matrices constitute the so called Gaussian Unitary Ensemble (GUE)  \cite{TraWid:09}.  We recall that a random variable $Z$ is said to have a standard complex Gaussian distribution (denoted ${\mathcal{CN}}(0, 1)$) if $Z = (Z_1 + i Z_2)$, where $Z_1$ and $Z_2$ are i.i.d. real Gaussian ${\mathcal N}(0, 1/2)$.

\medskip
\noindent
The following is a new Theorem for the GUE matrices.

\
\begin{theorem}
\label{th:cdfGUE}
The \ac{CDF} of the largest eigenvalue for the GUE is
\begin{equation}
\label{eq:cdfGUE}
F_{\lambda_1}(x_1)=\Pr\left\{\lambda_1 \leq x_1\right\}=K_{GUE} \left|{\bf A}(x_1)\right|
\end{equation}
where the elements of the $n \times n$ matrix ${\bf A}(x_1)$ are
\begin{eqnarray}
\label{eq:aijGUE}
a_{i,j}(x_1)&=&  \int_{-\infty}^{x_1} t^{i+j-2} e^{-t^2} dt  \\
&=&\frac{1}{2} \Gamma\left(\frac{i+j-1}{2} \right)   \left[\reggamma\left(\frac{i+j-1}{2}, x_1^2 \right)\sgn(x_1)^{i+j-1}+(-1)^{i+j}\right]  \nonumber
\end{eqnarray}
and $K_{GUE}=2^{n(n-1)/2} (\pi^{n/2} \prod_{i=1}^{n} \Gamma[i])^{-1}$ is a normalizing constant.
\end{theorem} 
\begin{proof}
For the GUE the joint distribution of the ordered eigenvalues can be written as \cite{TraWid:09}
\begin{equation} \label{eq:jpdfuncorrGUE}
f_{\bl{}} ({\bf x}) = K_{GUE} \left|{\bf V}({\bf x})\right|^2
     \prod_{i=1}^{n}e^{-x_{i}^2}
\end{equation}
Then, by using \cite[Th. 7]{ChiZan:C08} with $a=-\infty, b=x_1, \Psi_i(x_j)=\Phi_i(x_j)=x_j^{i-1}, \xi(x)=e^{-x^2}$ we get immediately the result. 
\end{proof}
%

\medskip

\subsection{Explicit distributions of the largest eigenvalue for finite dimensional Wishart and Gaussian matrices}
%
{Besides numerical computation of the \ac{CDF},} the previous theorems can be used to obtain explicit expressions for the distribution of the largest eigenvalue. 
%
For example, by expanding \eqref{eq:cdfwishartreal} we derive the following expressions for real Wishart matrices.

\noindent 
For $ \nmin=\nmax=2 $:
\begin{eqnarray}
\label{eq:cdf2x2}
F_{\lambda_1}(x)&=&\sqrt{\frac{x\pi }{2}} e^{-x/2} \text{erf}\sqrt{\frac{x}{2}}+e^{-x}-1  
%
%
\end{eqnarray}
For $ \nmin=2, \nmax=5$:
\begin{eqnarray}
\label{eq:cdf5x2}
F_{\lambda_1}(x)&=&\frac{1}{6} e^{- x} \left(2 e^{ x/2}  x^2+ x^2+6  x+6\right) -1
\end{eqnarray}
For $ \nmin=\nmax=3$:
\begin{eqnarray}
\label{eq:cdf3x3}
F_{\lambda_1}(x)&=&e^{-3 x/2} \left(e^{x/2} \left(e^x-x-1\right) \text{erf}\sqrt{\frac{x}{2}}-\sqrt{\frac{2x}{\pi}} \left(e^x (x-1)+1\right) \right)  
%
%
\end{eqnarray}
For $ \nmin=\nmax=4$:
{
\begin{eqnarray} 
\label{eq:cdf4x4}
F_{\lambda_1}(x)&=&\frac{e^{-2 x}}{\sqrt{32}}  \left(\sqrt{2} \left(4 e^{2 x}-e^x \left(x^3+2 x^2+2 x+8\right)+2 (x+2)\right) \right. \nonumber \\
&&\left. -\sqrt{\pi x} e^{x/2} \left(e^x \left(x^2-4 x+6\right)-2 (x+3)\right) \text{erf}\sqrt{\frac{x}{2}}\right)    
%
%
\end{eqnarray}
}
 
Similar expressions can be derived for the \ac{p.d.f.}, for complex Wishart, for GOE and for GUE. These expressions become cumbersome for large matrices.

\medskip

From a numerical point of view, with the previous expressions \eqref{eq:cdfwishartreal}, \eqref{eq:cdfGOE}, \eqref{eq:cdfwishartcomplex}, and \eqref{eq:cdfGUE}, which can be efficiently computed without numerical integration or series expansions, we can obtain the exact \ac{CDF} of the largest eigenvalue for matrices of large dimension. For example,  the numerical evaluation of \eqref{eq:cdfwishartreal} for Wishart real matrices with $\nmin=\nmax=500$ requires about $70$ seconds. 

\noindent 
If we need to work with larger matrices we can approximate the exact distributions with the limiting distributions described in the following sections.

\section{Limiting behavior for large random matrices: the Tracy-Widom distribution}

The pioneering works  \cite{TraWid:94,TraWid:96} and  \cite{Joha:00,Joh:01} have shown the importance of the  Tracy-Widom distribution, which arises in many fields as the limiting distribution of the largest eigenvalue of large random matrices. 
%
%
This distribution, originally derived in the study of the Gaussian unitary ensemble, has been shown to be related to many areas concerned with large random matrices.
Applications include \ac{PCA}, analysis of large data sets, combinatorics, communication theory, representation theory, probability, statistics and mathematical physics  \cite{Joha:00,Joh:01,Sos:02,Elk:07,TraWid:09,Nad:11}. 

For example, it has been shown that if $\mathbf{X}$ is an $n \times p$ matrix whose entries are \ac{i.i.d.} standard Gaussian and $\lambda_1$ is the largest eigenvalue of $\mathbf{X}\mathbf{X}^H$, then 
for $n, p \rightarrow \infty$  and $n/p \rightarrow \gamma  \in [0,\infty]$ 
\begin{equation}
\label{eq:l1pca}
\frac{\lambda_1 -  \mu_{np} }{\sigma_{np}} \overset{{\mathcal{D}}}{\longrightarrow} {\W_{\beta}}
\end{equation}
%
where $\W_{\beta}$ denotes a \ac{r.v.} with Tracy-Widom distribution of order $\beta$, for $\beta=1,2$ and $4$  \cite{Joha:00,Joh:01,Elk:03,TraWid:09}. In the previous expression $\beta=1$ when the entries of $\mathbf{X}$ are standard real Gaussian, and $\beta=2$ when the entries are standard complex Gaussian.  The case $\beta=4$ is of interest for the Gaussian Symplectic Ensemble (GSE) \cite{Meh:91}. 

\noindent 
The scaling and centering parameters in \eqref{eq:l1pca} are 
\begin{eqnarray}
\label{eq:munp}
\mu_{np}&=&\left( \sqrt{n+a_1}+\sqrt{p+a_2}\right)^2 \\
\label{eq:sigmanp}
\sigma_{np}&=&\sqrt{\mu_{np}} \left( \frac{1}{\sqrt{n+a_1}}+\frac{1}{\sqrt{p+a_2}}\right)^{1/3}
\end{eqnarray}
where the adjustment parameters $a_1, a_2$ are chosen here to be $a_1=a_2=-1/2$ for real Wishart ($\beta=1$) \cite{Joh:06,Ma:12} and $a_1=a_2=0$ for complex Wishart ($\beta=2$) \cite{Joh:01}. 
A similar behavior can be proved for more general conditions when the entries of $\mathbf{X}$ are not Gaussian \cite{Sos:02,Pec:09}. 
Due to the simplicity of this result, the Tracy-Widom distribution is of extreme usefulness for problems involving \ac{PCA} with large dimensional matrices.
 
The Tracy-Widom \acp{CDF} are given by \cite{TraWid:94,TraWid:96,TraWid:09}
\begin{equation}
\label{eq:F1}
F_1(x) =\exp\left\{-\frac{1}{2}\int_{x}^{\infty} q(y)+(y-x) q^2(y) dy\right\}
\end{equation}
\begin{equation}
\label{eq:F2}
F_2(x) =\exp\left\{-\int_{x}^{\infty} (y-x) q^2(y) dy\right\}
\end{equation}
\begin{equation}
\label{eq:F4}
F_4\left(\frac{x}{\sqrt{2}}\right) =\cosh\left\{\frac{1}{2}\int_{x}^{\infty} q(y) dy\right\} \sqrt{F_2(x)}
\end{equation}
where $q(y)$ is the unique solution to the Painlev\'{e} II differential equation 
\begin{equation}
\label{eq:pain}
q''(y) =y q(y)+2 q^3(y)
\end{equation}
satisfying the condition
\begin{equation}
q(y) \sim \ Ai(y) \qquad y \rightarrow \infty
\end{equation}
and $Ai(y)$ denotes the Airy function.

The function $F_4(x)$ can be derived from the other two since, from \eqref{eq:F1}, \eqref{eq:F2} and \eqref{eq:F4} we can write 
\begin{equation}
F_4\left({x}\right) =\frac{1}{2}\left(F_1(x\sqrt{2})+\frac{F_2(x\sqrt{2})}{F_1(x\sqrt{2})}\right)
\end{equation}
and 
\begin{equation}
f_4(x)=\frac{1}{\sqrt{2}}\left[f_1(x\sqrt{2})+\frac{f_2(x\sqrt{2})F_1(x\sqrt{2})-F_2(x\sqrt{2})f_1(x\sqrt{2})}{F_1^2(x\sqrt{2})}\right]
\end{equation}
where $f_{\beta}(x)=dF_{\beta}(x)/dx$. So in the following we will mainly focus on $F_1(x)$ and $F_2(x)$.

These distributions can be evaluated numerically by solving  the Painlev\'{e} II differential equation \eqref{eq:pain} or the corresponding Fredholm determinant \cite{TraWid:94,Joh:01,Elk:07,PraSpo:04,TraWid:09,Bor:10}. 
%

In this paper we propose a {new,} very simple approximation for the Tracy-Widom distribution, to avoid the need for numerical solution of differential equations {or} Fredholm determinants. The approximation is shown to be extremely accurate for values of the \ac{CDF} or of the \ac{CCDF} of practical uses. 
%
\section{A simple approximation of the Tracy-Widom distribution based on the gamma distribution}
\label{sec:simpappx}
\begin{table}[t]
\scriptsize
\caption{Parameters for approximating $\W_{\beta}$ with $\Gamma[\kappx,\tappx]-\aappx$.} 
\label{tab:par}
\begin{center}
\begin{tabular}{c l l l} \toprule
 &  $\W_1$ &  $\W_2$ &  $\W_4$ \\
 \midrule 
$ \kappx $&  46.446 & 79.6595 & 146.021 \\ 
 $\tappx$ & 0.186054 & 0.101037& 0.0595445\\
$ \aappx$ & 9.84801 & 9.81961 & 11.0016 \\
\bottomrule
\end{tabular}
\end{center}
\end{table}
We have proved that the exact distribution of the largest eigenvalue of a complex Wishart matrix is a mixture of gamma distributions. 
Actually, it can be verified that the mixture is generally well approximated by a gamma distribution, and relations of the extreme eigenvalues with the gamma distribution have been also reported in \cite[eq. (6.2)]{Ede:88}, \cite{EdeSut:05,CheDon:05}. In the same line, in \cite{WeiTir:11} the largest eigenvalue  for complex Wishart matrices is  approximated by a gamma, with proper parameters chosen to match the first two moments of the true distribution.

Since it has been proved that the largest eigenvalue distribution tends to the Tracy-Widom laws and it has been observed that the exact distribution is well approximated by a gamma distribution, we propose for the Tracy-Widom the approximation below.
\bigskip

{\it
\noindent
{\bf Approximation of the Tracy-Widom distribution.}

\noindent
The Tracy-Widom distribution can be accurately approximated by a scaled and shifted gamma distribution
\begin{equation}
\label{eq:appxmc} 
\W_{\beta}  \simeq \X-\aappx
\end{equation}
where $\aappx$ is a constant, and $\X \sim \Gamma(\kappx,\tappx)$ denotes a gamma \ac{r.v.} with shape parameter $\kappx$ and scale parameter $\tappx$. 
Thus the \ac{CDF} and \ac{p.d.f.} of $\W_{\beta}$ are approximated as:
\begin{eqnarray}
\label{eq:Fbetaappx}
F_{\beta}(x) &\simeq& 
  \frac{1}{\Gamma(\kappx)} \gamma\left(\kappx,\frac{x+\aappx}{\tappx}\right) \qquad\qquad\qquad x>-\aappx  \\ 
%
\label{eq:fbetaappx}
f_{\beta}(x) &\simeq& 
  \frac{1}{\Gamma(\kappx) \tappx^{\kappx}} \left(x+\aappx\right)^{\kappx-1} e^{-\frac{x+\aappx}{\tappx}}  \qquad\qquad x>-\aappx . 
\end{eqnarray}
%
%
}

We have chosen to set $\kappx,\tappx, \aappx$ for matching the first three moments of the distributions $\W_{\beta}$. 
To this aim we recall  that for the gamma \ac{r.v.} the {mean}  is $\EX{\X}=\kappx \tappx$, the variance is $\text{var}\left\{\X\right\}=\kappx \tappx^2$ and the skewness is $\text{Skew}\left\{\X\right\}=\frac{2}{\sqrt{\kappx}}$. If $\mu_{\beta}, \sigma_{\beta}^2, S_{\beta}$ are the mean, variance and skewness of the Tracy-Widom (see e.g. \cite{TraWid:09,PraSpo:04}), then matching the first three moments gives:
\begin{eqnarray}
\label{eq:matching}
\kappx&=&\frac{4}{S_{\beta}^2} \\
\tappx&=& \sigma_{\beta} \frac{S_{\beta}}{2}  \\ 
\aappx&=&\kappx \tappx-\mu_{\beta} 
\end{eqnarray}
The parameters for the approximation \eqref{eq:Fbetaappx},  \eqref{eq:fbetaappx} obtained from these equations are reported in Table \ref{tab:par}.

The comparison with pre-calculated \ac{p.d.f.} 
 values from \cite{PraSpo:04} is shown in Fig. \ref{fig:pdftw1tw2}. 
  Since in linear scale the exact and approximated distributions are practically indistinguishable, in 
   Fig. \ref{fig:logcdfccdftw2} we report the \ac{CDF} and \ac{CCDF} in logarithmic scale for Tracy-Widom 2 (similar for the others). It can be seen that the approximation is in general very good for all values of the \ac{CDF} of practical interest. 
In particular there is an excellent agreement between the exact and approximate distributions for the right tail. The left tail is less accurate but still of small relative error for values of the \ac{CDF} of practical statistical uses. 
Note that, differently from the true distribution which goes to zero only asymptotically, the left tail is exactly zero for $x<-\aappx$.
\begin{figure}[h]
\centerline{\includegraphics[width=0.6\columnwidth,draft=false]
    {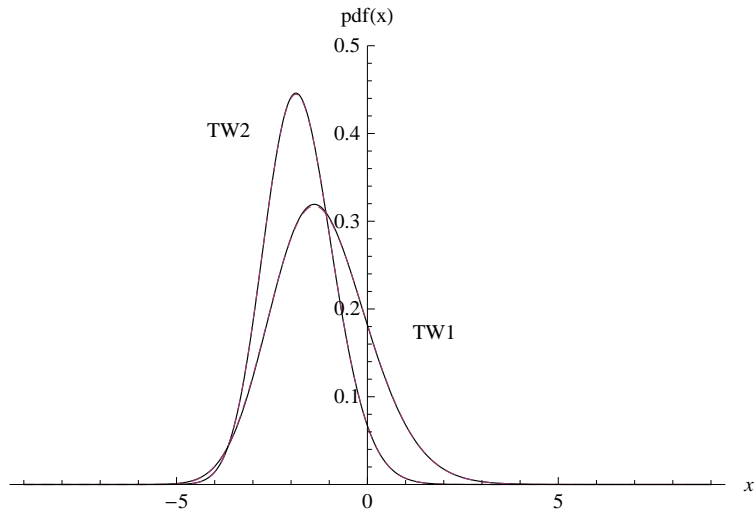}}
\caption{Comparison between the exact (solid line) and approximated (dashed) PDFs for the Tracy-Widom 1 and Tracy-Widom 2. The exact and approximated curves are practically indistinguishable on this scale.} \label{fig:pdftw1tw2}
\end{figure}

%

\begin{figure}[h]
\centerline{\includegraphics[width=0.6\columnwidth,draft=false]
    {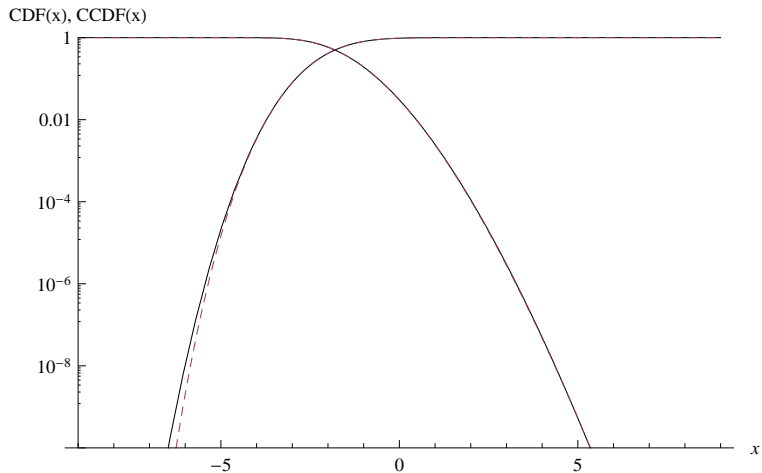}}
\caption{Comparison between the exact (solid line) and approximated (dashed) CDF, CCDF, Tracy-Widom ${\W_2}$, log scale. The two CCDFs are practically indistinguishable.} \label{fig:logcdfccdftw2}
\end{figure}

Some specific values are given in Table \ref{tab:tw1} and \ref{tab:tw2} 
 where it can be noted that, for values of common use, the relative error is small.
 
\begin{table}[h]
\scriptsize
\caption{Precision of the approximation: {CDF} of $\W_{1}$ vs. $\Gamma[\kappx,\tappx]-\aappx$ for some percentiles.} 
\label{tab:tw1}
\begin{center}
\begin{tabular}{c c c c c c}
\toprule
 $x$ &  Target {CDF} & {CDF}  \cite{PraSpo:04} & {CDF} approximation & {CDF} rel.~error (\%) &  {CCDF} rel.~error (\%) \\
 \midrule
 -4.64 & 0.001 & 0.0011 & 0.0009 & -17.40 & \ 0.02 \\
 -3.90 & 0.010 & 0.0099 & 0.0095 & -4.02 & \ 0.04 \\
 -3.18 & 0.050 & 0.0500 & 0.0501 & \ 0.16 & -0.01 \\
 -2.78 & 0.100 & 0.1004 & 0.1010 & \ 0.65 & -0.07 \\
 -1.91 & 0.300 & 0.3001 & 0.3011 & \ 0.32 & -0.14 \\
 -1.27 & 0.500 & 0.4995 & 0.4995 & -0.01 & \ 0.01 \\
 -0.59 & 0.700 & 0.7006 & 0.6998 & -0.12 & \ 0.27 \\
 \ 0.45 & 0.900 & 0.9000 & 0.8996 & -0.04 & \ 0.36 \\
 \ 0.98 & 0.950 & 0.9500 & 0.9500 & -0.00 & \ 0.01 \\
 \ 2.02 & 0.990 & 0.9899 & 0.9901 & \ 0.02 & -1.67 \\
 \ 3.24 & 0.999 & 0.9989 & 0.9990 & \ 0.01 & -4.96 \\
\bottomrule
\end{tabular}
\end{center}
\end{table}

\begin{table}[h]
\scriptsize
\caption{Precision of the approximation: {CDF} of $\W_{2}$ vs. $\Gamma[\kappx,\tappx]-\aappx$ for some percentiles.} 
\label{tab:tw2}
\begin{center}
\begin{tabular}{c c c c c c}
\toprule
 $x$ &  Target {CDF} & {CDF}  \cite{PraSpo:04} & {CDF} approximation & {CDF} rel.~error (\%) &  {CCDF} rel.~error (\%) \\
 \midrule
  -4.27 & 0.001 & 0.0011 & 0.0010 & -8.50 & \ 0.01 \\
 -3.72 & 0.010 & 0.0102 & 0.0100 & -1.77 & \ 0.02 \\
 -3.19 & 0.050 & 0.0505 & 0.0506 & \ 0.16 & -0.01 \\
 -2.90 & 0.100 & 0.1003 & 0.1006 & \ 0.35 & -0.04 \\
 -2.26 & 0.300 & 0.3025 & 0.3029 & \ 0.14 & -0.06 \\
 -1.80 & 0.500 & 0.5022 & 0.5021 & -0.02 & \ 0.02 \\
 -1.32 & 0.700 & 0.7018 & 0.7014 & -0.06 & \ 0.14 \\
 -0.59 & 0.900 & 0.9012 & 0.9011 & -0.02 & \ 0.16 \\
 -0.23 & 0.950 & 0.9503 & 0.9503 & \ 0.00 & -0.03 \\
\ 0.48 & 0.990 & 0.9901 & 0.9901 & \ 0.01 & -0.91 \\
\ 1.31 & 0.999 & 0.9990 & 0.9990 & \ 0.00 & -2.63 \\
%
\bottomrule
\end{tabular}
\end{center}
\end{table}


\subsection{Approximating the {CDF} of the largest eigenvalue for Wishart and Gaussian matrices}

Thus, putting the gamma approximation in \eqref{eq:l1pca} we have, for Wishart matrices 
with $n, p \rightarrow \infty$, %
\begin{equation}
\label{eq:l1pcagamma}
\frac{\lambda_1 -  \mu_{np} }{\sigma_{np}} + \alpha\overset{{\mathcal{D}}}{\approx} {\Gamma[k,\theta]}
\end{equation}
%
where $ \mu_{np}, \sigma_{np}$ are those in \eqref{eq:munp}, \eqref{eq:sigmanp}, $\Gamma[k,\theta]$ is the gamma distribution and $k, \theta, \alpha$ are given in Table \ref{tab:par}.  A similar approximation has been used in \cite{WeiTir:11} for complex Wishart matrices, where $\lambda_1$ is approximated by a gamma distribution, and the first two moments of the Tracy-Widom are used to find the two parameters of the gamma. Here we have a different approach, since we approximate the Tracy-Widom with a shifted gamma, thus $\lambda_1$ is also a shifted gamma. 

Similarly,  the  approximation based on the Tracy-Widom distribution for GOE and GUE is \cite{TraWid:94,TraWid:96,TraWid:09}: 
\begin{equation}
\label{eq:l1GOEGUE}
\frac{\lambda_1 -  \mu'_{n} }{\sigma'_{n}} \overset{{\mathcal{D}}}{\longrightarrow} {\W_{\beta}}
\end{equation}
with $\mu'_{n}=2\sigma_0 \sqrt{n-a_1}$ and $\sigma'_{n}=\sigma_0 (n-a_2)^{-1/6}$, where $\sigma_0^2=1/2$ is the variance of the off-diagonal elements in the ensembles in our normalization. In the previous expression $\beta=1$ and $\beta=2$ for the GOE and GUE, respectively. 
For the GOE, following  \cite{JohMa:12} we used 
$a_1=1/2+1/10 (n - 1/2)^{-1/3}$,  $a_2=0$. 
For the GUE, one possible choice is $a_1=a_2=0$ \cite{TraWid:09}, but we have observed that the approximations are better for small $n$ with $a_1=0, a_2=1$.

%

\noindent Thus, 
for large $n$ we have for GOE and GUE the new expression:
\begin{equation}
\label{eq:l1GOEGUEgamma}
\frac{\lambda_1 -  \mu'_{n} }{\sigma'_{n}}  + \alpha\overset{{\mathcal{D}}}{\approx} {\Gamma[k,\theta]}
\end{equation}
In the following section, formulas \eqref{eq:l1pcagamma} and \eqref{eq:l1GOEGUEgamma} are compared with the exact distributions provided by Theorems 1-4.

\section{Numerical results}
The calculation of  the exact distribution of the largest eigenvalue is easy by using Theorems 1-4 for not too large random matrices (e.g., Wishart matrices with $\nmin=500$). 
For example, we show in Fig. \ref{fig:cdfwishartreal5}, Fig. \ref{fig:cdfwishartreal500} and Fig. \ref{fig:cdfwishartcomplex5} 
 the distribution of the largest eigenvalue for real and complex Wishart matrices, with $\nmin=5$ and $\nmin=500$. In the figure we report the  exact distributions given by  \eqref{eq:cdfwishartreal}, \eqref{eq:cdfwishartcomplex} and the centered and scaled Tracy-Widom distribution \eqref{eq:l1pca} (here we can use the exact Tracy-Widom or the approximations \eqref{eq:Fbetaappx} which are non distinguishable in this scale).

In Fig. \ref{fig:cdfGOE} and Fig. \ref{fig:cdfGUE} we report the exact distribution for GOE (eq. \eqref{eq:cdfGOE}) and for GUE (eq. \eqref{eq:cdfGUE}), for $n=2,5,10,20,50$. In the same figures we report the  approximation based on the Tracy-Widom distribution. 
%

We note that, for large dimension problems, the asymptotic distributions predicted by the Tracy-Widom laws converge soon to the exact. 
In particular, for GOE and GUE the properly scaled and centered Tracy-Widom laws are already very close to the exact for very small matrices ($n=2$). Also, we remark that the simple approximations \eqref{eq:Fbetaappx}, \eqref{eq:fbetaappx} can be used instead of the pre-calculated tables for the Tracy-Widom distribution for values of practical interest in statistic.

\section*{Acknowledgements}
This work was partially supported by the Italian Ministry of Education, Universities and Research (MIUR) under the PRIN Research Project ``GRETA''.

The author would like to thank Massimo Cicognani, Andrea Mariani, Moe Z. Win, and Alberto Zanella for discussions and comments. Thanks to Raymond Kan for careful comments on an earlier version of the paper and for suggesting the iteration \eqref{eq:aij+1}. The constructive comments of the Reviewers and the Editor contributed to strengthen and widen this work.


%
\bibliographystyle{biometrika}

\bibliography{RandomMatrix,IEEEabrv,MyBooks,BiblioMCCV,MIMO,BibEIGEN}

\begin{figure}
\centerline{\includegraphics[width=0.6\columnwidth,draft=false] {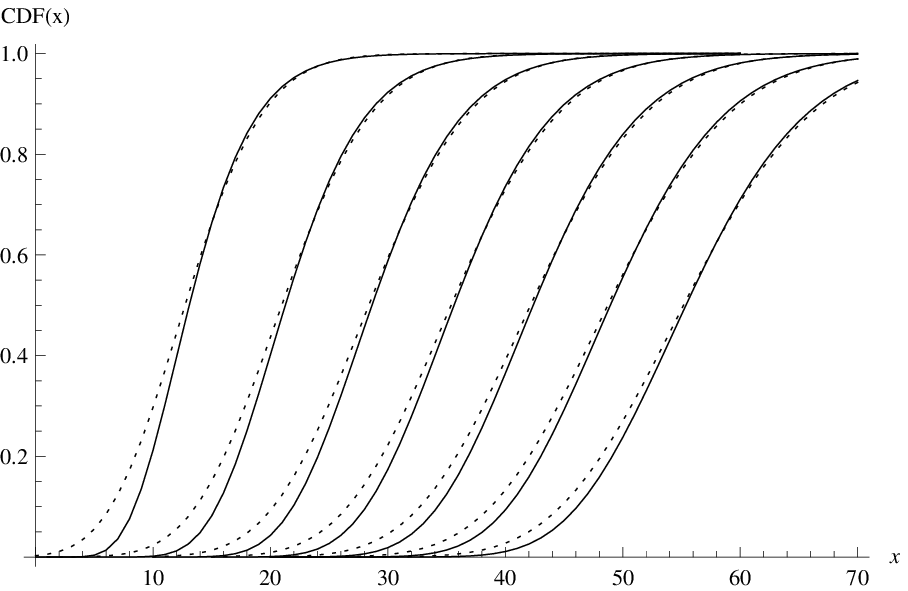}}
\caption{{CDF} of the largest eigenvalue, real Wishart matrix, $\nmin=5, \nmax=5,10,15,20,25,30,35$. Comparison between the exact distribution \eqref{eq:cdfwishartreal} (solid line) and the scaled and centered ${\W_1}$ as in \eqref{eq:l1pca} (dotted). } \label{fig:cdfwishartreal5}
\end{figure}


\begin{figure}
\centerline{\includegraphics[width=0.6\columnwidth,draft=false] {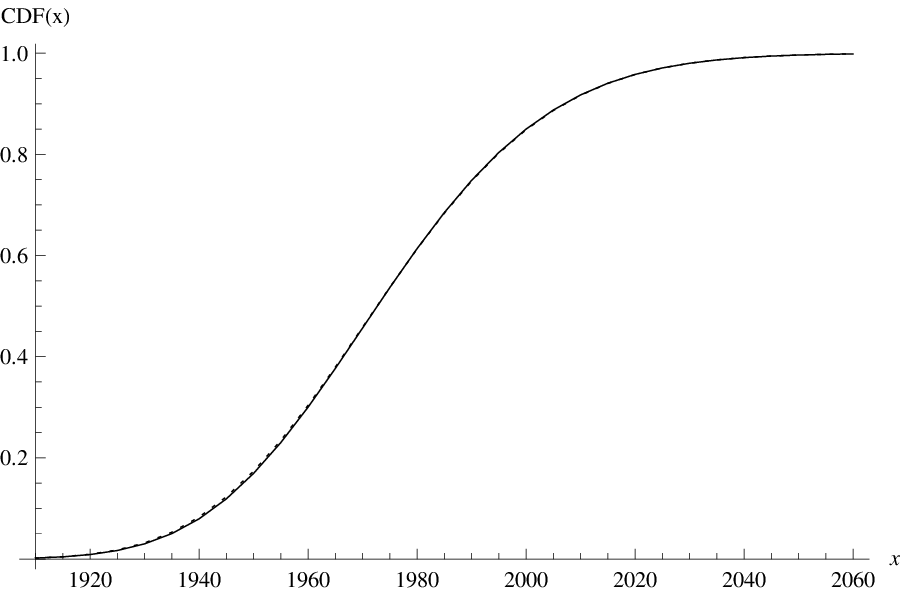}}
\caption{{CDF} of the largest eigenvalue, real Wishart matrix, $\nmin=500, \nmax=500$. Comparison between the exact distribution \eqref{eq:cdfwishartreal} (solid line) and the scaled and centered ${\W_1}$ as in \eqref{eq:l1pca} (dotted). } \label{fig:cdfwishartreal500}
\end{figure}



\begin{figure}
\centerline{\includegraphics[width=0.6\columnwidth,draft=false] {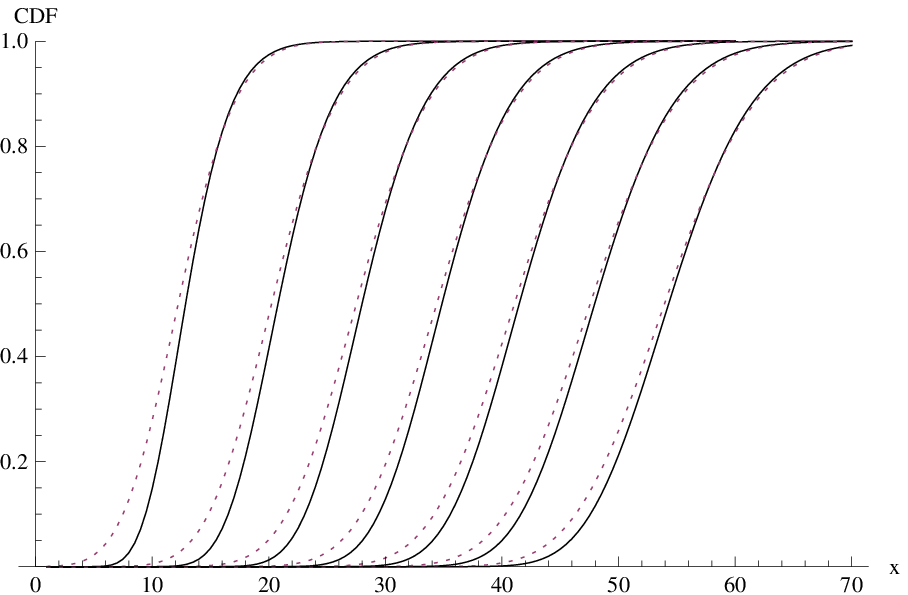}}
\caption{{CDF} of the largest eigenvalue, complex Wishart matrix, $\nmin=5$, $\nmax=5, 10, 15, 20, 25, 30, 35$. Comparison between the exact distribution \eqref{eq:cdfwishartcomplex} (solid line) and the scaled and centered ${\W_2}$ as in \eqref{eq:l1pca} (dotted). } \label{fig:cdfwishartcomplex5}
\end{figure}


\begin{figure}
\centerline{\includegraphics[width=0.6\columnwidth,draft=false] {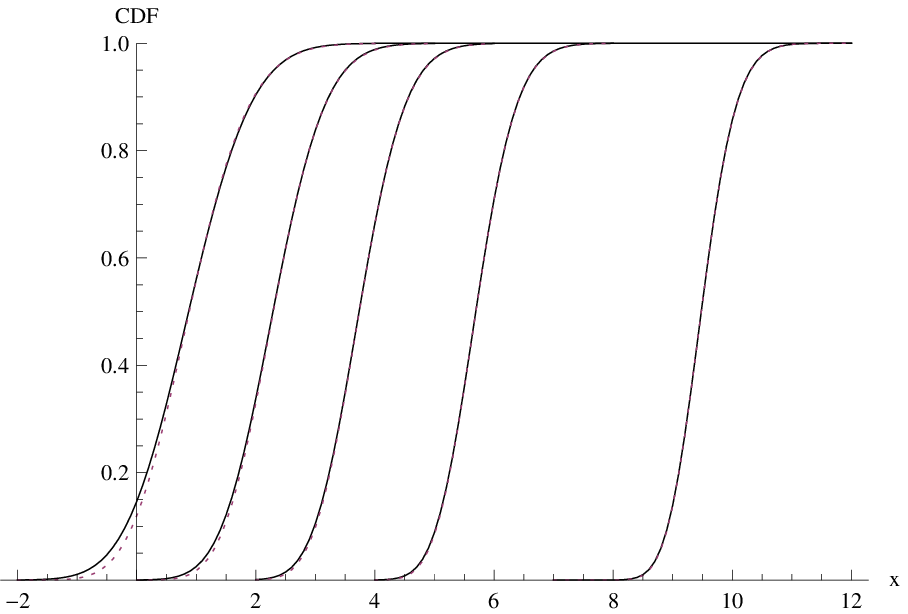}}
\caption{{CDF} of the largest eigenvalue, GOE. From left to right: $n=2, 5, 10, 20, 50$. Comparison between the exact distribution \eqref{eq:cdfGOE} (solid lines) and the scaled and centered ${\W_1}$ as in \eqref{eq:l1GOEGUE} (dotted lines). 
} 
\label{fig:cdfGOE}
\end{figure}

\begin{figure}
\centerline{\includegraphics[width=0.6\columnwidth,draft=false] {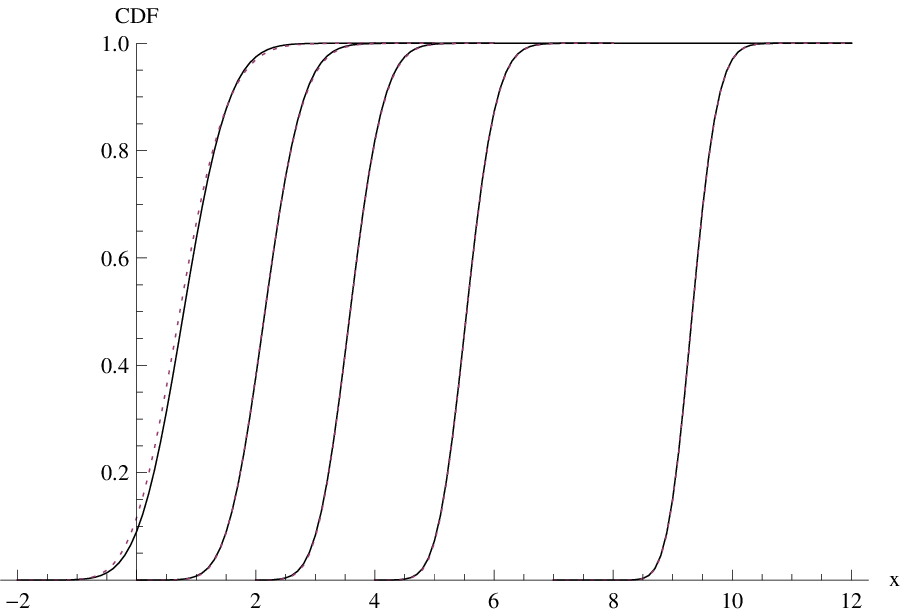}}
\caption{{CDF} of the largest eigenvalue, GUE. From left to right: $n=2, 5, 10, 20, 50$. Comparison between the exact distribution \eqref{eq:cdfGUE} (solid lines) and the scaled and centered ${\W_2}$ as in \eqref{eq:l1GOEGUE} (dotted lines). 
}
 \label{fig:cdfGUE}
\end{figure}

\end{document}